\documentclass[10pt,a4paper,envcountsame]{llncs}
\usepackage[utf8]{inputenc}
\usepackage{amsmath}
\usepackage{amsfonts}
\usepackage{amssymb}
\usepackage{thmtools,thm-restate}
\usepackage[textsize=small]{todonotes}
\usepackage{paralist}
\usepackage{xspace}
\usepackage{hyperref}
\usepackage[capitalize]{cleveref}

\newcommand{\ignore}[1]{}

\newcommand{\defstyle}[1]{\emph{#1}}

\newcommand{\A}{\mathcal{A}}
\newcommand{\B}{\mathcal{B}}

\renewcommand{\O}{\mathcal{O}}
\newcommand{\N}{\mathbb{N}}
\newcommand{\Z}{\mathbb{Z}}

\newcommand{\ptime}{\textsc{PTime}\xspace}
\newcommand{\nctwo}{\textsc{NC}$^2$\xspace}
\newcommand{\pspace}{\textsc{PSpace}\xspace}
\newcommand{\nl}{\textsc{NL}\xspace}
\newcommand{\polylogspace}{\textsc{PolyLogSpace}\xspace}
\newcommand{\logspace}{\textsc{L}\xspace}
\newcommand{\nlogspace}{\textsc{NL}\xspace}
\newcommand{\np}{\textsc{NP}\xspace}
\newcommand{\conp}{\textsc{coNP}\xspace}
\newcommand{\expspace}{\textsc{ExpSpace}\xspace}
\newcommand{\twoexptime}{\textsc{2ExpTime}\xspace}
\newcommand{\twoexpspace}{\textsc{2ExpSpace}\xspace}

\newcommand{\ackermann}{\textsc{Ackermann}\xspace}

\newcommand{\set}[1]{\{#1\}}

\newcommand{\fin}{\textsc{fin}}

\newcommand{\trans}[1]{\xrightarrow{#1}}

\newcommand{\init}{\textup{init}}

\newcommand{\eps}{\varepsilon}
\renewcommand{\epsilon}{\varepsilon}
\newcommand{\ie}{\emph{i.e.}}

\newcommand{\dcup}{\mathop{\dot\cup}}

\newcommand{\partition}{\textsc{Perfect Partition}\xspace}

\newcommand{\lang}{\mathcal{L}}

\renewcommand{\enspace}{}

\title{Universality Problem for Unambiguous VASS} 

\titlerunning{Universality Problem for Unambiguous VASS}

\author{Wojciech Czerwi\'nski\inst{1}
\and Diego Figueira\inst{2}
\and Piotr Hofman\inst{1}}

\institute{University of Warsaw \and Univ.\ Bordeaux, CNRS,  Bordeaux INP, LaBRI}

\begin{document}
\maketitle
\begin{abstract}
We study languages of unambiguous VASS, that is, Vector Addition Systems with States, whose transitions read letters from a finite alphabet, and whose acceptance condition is defined by a set of final states (\ie, the coverability language). We show that the problem of universality for unambiguous VASS is \expspace-complete, in sheer contrast to \ackermann-completeness for arbitrary VASS, even in dimension $1$. When the dimension $d \in \N$ is fixed, the universality problem is \pspace-complete if $d \geq 2$, and $\conp$-hard for $1$-dimensional VASSes (also known as {\sl One Counter Nets}).
\end{abstract}



\section{Introduction}
Determinism is a central notion of computational models, it ensures that there is one way to proceed for every input. It often enables constructions which would not be possible without it and allows for efficient algorithms. While the relation between deterministic vs non-deterministic models is extensively studied, there exists also a less understood middle ground of {\sl unambiguous} systems. In the case of models accepting word languages, a model is said to be \emph{unambiguous} if for every word in its language, there is exactly one accepting run, which is a much weaker restriction than determinism.
Unambiguity, although featuring non-determinism, often causes some problems to be computationally easier. As a prominent example, the universality problem for finite automata (\ie, whether all words over the alphabet are accepted by the automaton), which is \pspace-complete in general, is known to be in \ptime in the unambiguous case~\cite{DBLP:journals/siamcomp/StearnsH85}
and even in \nctwo~\cite{DBLP:journals/ipl/Tzeng96}.
While the study of unambiguous models of computation has lately attracted some attention, in some settings it remains, by and large, an unexplored area.

In particular, there has been considerable volume of research on unambiguous {\sl finite automata} (see~\cite{DBLP:conf/dcfs/Colcombet15} for a nice overview).
One way to design a polynomial time algorithm for the universality problem on finite automata is to show that
the shortest word which is not in the language, if any, is of at most linear length. Then, by counting the number of 
linear length runs one may answer the problem. The existence of a linear counterexample for universality and its \ptime algorithm, led to the conjecture, formulated by Colcombet~\cite{DBLP:conf/dcfs/Colcombet15}, that for every unambiguous finite automaton (UFA)
there exists another UFA of polynomial size accepting the complement of its language.
This conjecture was later shown false by Raskin~\cite{DBLP:conf/icalp/Raskin18}. As it turns out, there is a family of UFA
such that for accepting the complement of UFA with $n$ states even nondeterministic finite automaton (NFA) needs
a super-polynomial number of states ---at least $\Theta(n^{\log \log \log n})$.
The universality problem for UFA is actually known to be not only in \ptime, but even in \nctwo~\cite{DBLP:journals/ipl/Tzeng96},
the class of problems solvable by uniform families of circuits with $\O(\log^2 n)$ depth and binary fan-in.
The work~\cite{DBLP:journals/ipl/Tzeng96} in fact solves the more general problem of {\sl path equivalence} for two NFA: is the number of accepting runs on $w$ the same for both automata, for every word $w$? However, to the best of our knowledge the best known lower bound for the problem is \nl-hardness, so the exact complexity of universality problem for UFA is still open
even in the simplest possible setting of finite automata.

There was also research about the universality problem and related ones for unambiguous register automata.
In~\cite{DBLP:conf/stacs/MottetQ19} authors have shown that the containment problem for unambiguous register automata is in \twoexpspace
and even in \expspace if the number of registers is fixed, which implies similar upper bounds for the universality problem.
Without the unambiguity assumption, even the universality problem (and even with just one register) can be shown undecidable~\cite{DBLP:journals/tocl/NevenSV04} or Ackermann-hard~\cite{DBLP:journals/tocl/DemriL09} depending on the concrete model of register automata.

It is not by accident that existing research focuses on universality, equivalence and containment of languages of unambiguous systems, and that there are efficient algorithms for these problems under the assumption of unambiguity.
Unambiguity speaks about the language of a system, so it is natural to hope that problems related to the language
of the systems may become more tractable. But for the most natural problem concerning the language, i.e., for the emptiness problem
one cannot hope for improvement. This is because for most of the systems one can relabel transitions giving each one a unique label.
Then the system becomes deterministic and in consequence unambiguous. The language changes, but it is empty iff the original language was empty,
which intuitively explains why the emptiness problem shouldn't be any easier for unambiguous systems compared to general non-deterministic ones.
On the other hand, it is more reasonable to expect that the universality problem might be easier since both the universality problem and the unambiguity property are universal properties of the form ``\textsl{For all words, [...]}''.

\paragraph*{Our contribution} The foremost goal of this paper is to push the understanding of unambiguity further. We focus on the universality problem, which is arguably the most natural first step, that may open the way for further studies on the equivalence, co-finiteness, containment and other problems for languages.
The universality problem was studied for finite automata and register automata under the unambiguity assumption.
In our opinion, the most interesting yet unsolved cases in which one can expect some progress assuming unambiguity are One Counter Nets
(called also 1-dimensional VASS here) and its generalization Vector Addition Systems with States (VASS).

The universality checking for VASS with state acceptance is known to be decidable by the use of well quasi-order techniques~\cite{DBLP:journals/jcss/JanarEM99}
(the paper shows decidability of trace universality, but language universality can be reduced to that problem).
However the problem is also known to be \ackermann-complete even for 1-dimensional VASS~\cite{DBLP:conf/rp/HofmanT14}, so hardly tractable.
For deterministic VASS it is quite easy to show that the universality problem can be decided in \ptime.
Therefore, it is natural to hope for improvement
under the unambiguity restriction.

Our main contribution is \expspace membership of the universality problem for unambiguous VASS.
We believe that it is the most interesting result and it was as well the most challenging problem and technically involved solution.
We actually have shown that this problem is \expspace-complete.
For the completeness of the picture we have also analyzed the complexity of the problem for $d$-dimensional VASS for fixed $d \in \N$.
We have shown that the problem is \pspace-complete for every $d \geq 2$.
For $d = 1$ we have shown \conp-hardness,
although we do not have the matching upper bound, we conjecture that it is \conp-complete.
We additionally consider the variant of the problem in which the numbers in the input are encoded in unary. Finally, we study also the problem of unambiguity checking  (\ie, given a VASS, is it unambiguous?).
All our results are listed in Section~\ref{sec:results}.


\section{Preliminaries}
We use the letter $\Sigma$ to denote a finite alphabet, $\Z$ to denote the set of all integers, and $\N$ the set of non-negative integers.
We use $\epsilon$ to denote the empty string, and $\Sigma_\epsilon$ to denote $\Sigma \cup \set{\epsilon}$.
\newcommand{\subsetfin}{\subseteq_{\mathit{fin}}}%
\newcommand{\wpfin}{\wp_{\mathit{fin}}}%
We use $A\subsetfin B$ to denote that $A$ is a finite subset of $B$, and $\wpfin(A)$ to denote the set of all finite subsets of $A$. We use $\bar u, \bar v, \bar w, \dotsc$ to denote vectors of numbers, and we use $\bar 0$ to denote the all-$0$ vector and $\bar 1$ to denote the all-$1$ vector.
We use $[i,j]$ for $i, j \in \N$, $i \leq j$ to denote the set $\{i, i+1, \ldots, j-1, j\}$.
For a vector $\bar u \in \Z^d$ and $i \in [1,d]$ we denote by $\bar u[i]$ the $i$-th coordinate of $\bar u$.
For a word $w \in \Sigma^*$ and $i \in \N$, $i > 0$ we denote by $w[i]$ the $i$-th letter of $w$.
For $\bar u, \bar v \in \Z^d$ we write $\bar u \preceq \bar v$ if for all $i \in [1,d]$ we have $\bar u[i] \leq \bar v[i]$.
We define the minimum of $\bar u$ and $\bar v$ as $\min(\bar u, \bar v)[i] = \min(\bar u[i], \bar v[i])$ for any $i \in [1,d]$.

\newcommand{\run}{\xrightarrow}%
\newcommand{\prun}{\xdashrightarrow}%
We consider a Vector Addition Systems with States (VASS) of dimension $d \in \N$ as a tuple $\A = (\Sigma,d,Q,q_0,\delta,F)$  where $\Sigma$ is a finite alphabet, $Q$ is a finite state space,  $q_0 \in Q$ is the  initial state, $F \subseteq Q$ is the set of final states, and $\delta \subsetfin Q \times \Sigma_\epsilon \times \Z^d \times Q$ is the set of transitions.
We often write transition $(p, a, v, q)$ as $p \trans{a; v} q$. We will henceforth write \defstyle{$d$-VASS} to denote a VASS of fixed dimension $d$.
A \defstyle{configuration} of $\A$ is a pair of a state $q \in Q$ and a vector $\bar u \in \N^d$, that we usually note $q(\bar u)$. If $c$ is a configuration, we write $c[i]$ to denote the $i$-th coordinate of the vector it contains. A \defstyle{run} of $\A$ from a configuration $q(\bar u)$ to a configuration $q'(\bar v)$ reading the word $w \in \Sigma^*$ is a sequence of transitions
$(r_1,\alpha_1,\bar v_1,r'_1) \dotsb (r_n,\alpha_n,\bar v_n,r'_n) \in \delta^*$ such that: (i) $r_1 = q$ and $r'_n = q'$, (ii) $r'_i = r_{i+1}$ for every $1 \leq i < n$; (iii) $w = \alpha_1 \dotsb \alpha_n$; (iv) $\bar u + \sum_{i \leq j} \bar v_i \in \N^k$ for every $1\leq j \leq n$; and (v) $\bar v = \bar u + \sum_{i \leq n} \bar v_i$.
If we further have $q' \in F$, we say that such run is \defstyle{accepting}.
We henceforth say that a configuration $c$ is \defstyle{reachable} from a configuration $c'$ if there is a run from $c'$ to $c$.
The \defstyle{effect} of a transition $(r, \alpha, \bar v, r')$ is the vector $\bar v \in \Z^d$, the \defstyle{effect of a run} is the sum of effects of the transitions therein.
%
The \defstyle{norm} of a VASS $\A$ is the maximal absolute value of a number occurring in its transition, and we denote it by $|\A|$.
The \defstyle{language} of a configuration $c$ in $\A$, denoted by $\lang(\A, c)$, is the set of all $w \in \Sigma^*$ with an accepting run from $c$.
We call $q_0(\bar 0)$ the \defstyle{initial configuration} where $q_0$ is the initial state.
If $c$ is the initial configuration then we just say language of $\A$ and write $\lang(\A)$ instead of $\lang(\A, c)$.
A VASS $\A$ is \defstyle{unambiguous} if for every $w \in \Sigma^*$ there is no more than one accepting run starting from the initial configuration and reading $w$. The \defstyle{unambiguity checking problem} for VASS is the problem of, given a VASS $\A$, decide whether it is unambiguous.
An automaton over $\Sigma$ (finite automaton or VASS) is \defstyle{universal} if it accepts the language $\Sigma^*$. The \defstyle{universality problem} for VASS is the problem of, given a VASS $\A$, decide whether it is universal.
We will henceforth assume that the numbers contained in the transitions of VASSes are always encoded in binary if not explicitly indicated otherwise. 

Observe that we work with VASS with $\eps$-transitions, the reason for doing so is that it is a natural model, the upper bounds still hold in this more general setup, and we can also derive tight lower bounds by making use of $\eps$-transitions. We do not know whether adding $\eps$-transitions increases the class of recognized languages, not even in the non-deterministic case. It seems to us a rather difficult question.

Let us recall now the main result of the Rackoff construction~\cite{DBLP:journals/tcs/Rackoff78}.
Let us denote $A_{M,d,n} = (2n^2(M+1)^2)^{(4d)^{d-1}}$. We present here an adaptation of the Rackoff argument
with an explicit bound on the length of an accepting run.

\begin{proposition}[Adaptation of the Rackoff construction]\label{prop:rackoff}
If a language of a $d$-VASS with norm $M$ and $n$ states is nonempty
then there exists an accepting run of length at most $A_{M,d,n}$.
\end{proposition}
\begin{proof}
	Let $C = 2n^2(M+1)^2$.
	We proceed by induction on $d$. For $d = 1$ assume there is some accepting run with no configuration repeating.
	Then in its prefix of length $nM$ there is definitely first a configuration $q(x)$ and later a configuration $q(\bar y)$ for some state
	$q$ and counter values $x < y$. Then we can change this accepting run into an accepting run of length at most $nM + (nM)^2 + n-1$.
	We first pump the infix from $q(x)$ do $q(y)$ exactly $nM$ times obtaining then a configuration $q(z)$ with $z = y + nM(y-x) \geq nM$.
	As some accepting state is reachable from $q$ then it is also reachable by a run of length smaller than $n$. This run
	(and any of its prefixes) can, at worst, have a negative effect of value $(n-1)M$, and thus it can be triggered from $q(z)$, since $z \geq nM$.
	In this way, we get an accepting run of length at most $nM + (nM)^2 + n-1 \leq C$, proving the  base case.

	For the inductive step, assume that there is an accepting run $s(\bar 0) \trans{\rho} f(\bar v)$ in a $(d+1)$-VASS with norm $M$ and $n$ states.
	Let $K_d = C^{(4d)^{d-1}}$. We distinguish two cases:
	\begin{enumerate}[(i)]
	  \item the norm of every configuration on $\rho$ is bounded by $C \cdot K_d$;
	  \item the norm of some configuration on $\rho$ exceeds $C \cdot K_d$.
	\end{enumerate}
	Without loss of generality we can assume that no configuration on $\rho$ appears more than once, otherwise we can ``unpump'' $\rho$ to obtain a shorter one.
	Observe that in the first case (i), the length of $\rho$ is bounded by $D = (C \cdot K_d)^{d+1}$ (we will bound $D$ later on).

	In the second case (ii), the run $\rho$ might be long, but we will show that there is another short accepting run $\rho'$.
	Let $p(\bar u)$ be the first configuration on $\rho$ with norm exceeding $C \cdot K_d$.
	Let $s(\bar 0) \trans{\rho_1} p(\bar u) \trans{\rho_2} f(\bar v)$. Clearly, the length of $\rho_1$ is bounded by $D$ by a similar reasoning as in the case (i).
	We will replace $\rho_2$ with a ``short'' run $\pi$, so that $c \trans{\pi} f(\bar v')$. First note that  some coordinate of $p(\bar u)$ must have value greater or equal to $C \cdot K_d$; without loss of generality, assume it is the last one, that is, the $(d+1)$-st coordinate.
	Let us now ignore the last coordinate in the VASS. By inductive hypothesis, there is a sequence of transitions $\pi$
	of length at most $K_d$ such that  $p(\bar u_d) \trans{\pi_d} f(\bar v'_d)$, where $\pi_d$ is the result of ignoring the last coordinate of $\pi$, and $\bar u_d, \bar v_d \in \N^d$ are the results of ignoring the last coordinate of $\bar u,\bar v$.
	%
	%
	Consider now the sequence of transitions $\pi$ starting in $p(\bar u)$. Its length is bounded by $K_d$, so its effect on the $(d+1)$-st coordinate
	is not smaller than $-M \cdot K_d$. 
	Since $\bar u[d+1] \geq C \cdot K_d \geq M \cdot K_d$, then  $\pi$ is indeed a valid run from $p(\bar u)$ to $f(\bar v')$ for some $\bar v' \in \N^{d+1}$.
	Therefore, the run $\rho_1 \cdot \pi$ is accepting from $s(\bar 0)$ as $s(\bar 0) \trans{\rho_1} p(\bar u) \trans{\pi} f(\bar v')$.
	The length of $\rho_1 \cdot \pi$ is at most $D + K_d$.

	In order to finish the argument in case (ii) we need to show that $D + K_d \leq K_{d+1}$, through the following sequence of (very rough) estimations
	\begin{align*}
	D + K_d & \leq 2D \leq C \cdot D = C \cdot (C \cdot K_d)^{d+1} = C \cdot \big( (C \cdot C^{(4d)^{d-1}})^{d+1} \\
	& = C^{((4d)^{d-1}+1)(d+1)+1} \leq C^{4(4d)^{d-1} \cdot (d+1)} \leq C^{(4(d+1))^d} = K_{d+1}.
	\end{align*}
	Observe that in case (i), the bound $D \leq K_{d+1}$ is trivial.
\end{proof}

The language emptiness problem for VASS (\ie, given a VASS, does it accept at least one word?) is, basically, equivalent to the coverability problem, which is known to be \expspace-complete as shown by the lower bound of Lipton~\cite{lipton1976} and the upper-bound of Rackoff~\cite{DBLP:journals/tcs/Rackoff78}. The coverability problem is the problem of, given a VASS $\A$ and two configurations $c_1,c_2$,  whether there is a run from $c_1$ to some configuration $c'_2$ such that $c'_2\succeq c_2$. In our setting, this result can be restated as the language emptiness problem for VASS being \expspace-complete, even when all transitions are $\epsilon$-transitions, and hence the language is either $\emptyset$ or $\set{\epsilon}$. 
What is more, the construction of Lipton is unambiguous: if there is an accepting run, there is exactly one. Indeed, the only situation in which Lipton's construction is ambiguous along a run is when it guesses whether the value of some counter is zero or non-zero. However, the run of a wrong guess is never an accepting one, as the guess is always followed by a verification.
This is formalized in the next lemma. Let us denote by $\epsilon$-VASS, a VASS whose every transition reads $\epsilon$ (and thus the alphabet is not important here).

\begin{lemma}[consequence of \cite{lipton1976,DBLP:journals/tcs/Rackoff78}]\label{lem:epsVASS-emptiness}
The problem of whether an unambiguous $\epsilon$-VASS has an empty language is \expspace-complete. 
\end{lemma}


\section{Results}\label{sec:results}
We summarize all our results in the next two theorems. Detailed proofs will come in the sections that follow.

\begin{theorem}\label{thm:universality}
The universality problem for
\begin{enumerate}[(i)]
  \item \label{it:vass-expspace} VASS is \expspace-complete, both with binary and unary encodings;
  \item \label{it:d-vass-unary-nctwo} $d$-VASS with unary encoding is in \nctwo and \nlogspace-hard, for every $d \geq 1$;
  \item \label{it:d-vass-binary-pspace} $d$-VASS with binary encoding is \pspace-complete, for every $d \geq 2$;
  \item \label{it:ocn-binary-conp} $1$-VASS (One Counter Net) with binary encoding is \conp-hard.
\end{enumerate}
\end{theorem}

\begin{theorem}\label{thm:testing-unambiguity}
The unambiguity checking problem for
\begin{enumerate}[(i)]
  \item \label{it:unambcheck:vass-expspace} VASS is \expspace-complete, both with binary and unary encodings;
  \item \label{it:unambcheck:d-vass-logspace} $d$-VASS with unary encoding is \nlogspace-complete, for every $d \geq 1$;
  \item \label{it:unambcheck:d-vass-bin-pspace} $d$-VASS with binary encoding is \pspace-complete, for every  $d \geq 2$;
  \item \label{it:unambcheck:ocn-bin-conp} $1$-VASS with binary encoding is \conp-hard.
\end{enumerate}
\end{theorem}

The main technical contribution lies in the \expspace bounds on the universality problem in Theorem~\ref{thm:universality}\eqref{it:vass-expspace}. The upper bound will need some insights on the structure of accepting runs in unambiguous VASS which happen to have a universal language. The remaining upper bounds will follow easily from this one. The \pspace, and \conp lower bounds of items
\eqref{it:d-vass-binary-pspace}, and \eqref{it:ocn-binary-conp} are also of interest, as they reveal different ways in which unambiguity can encode non-trivial properties.
The \expspace lower bound of item \eqref{it:vass-expspace} follows easily from Lemma~\ref{lem:epsVASS-emptiness}.
All the remaining results of Theorems~\ref{thm:universality} and \ref{thm:testing-unambiguity} are either easy, or follow from simple adaptations of the three results just mentioned.

It is interesting to observe that complexity results on universality seem to coincide with the complexity of emptiness for the non-deterministic version of the considered classes. Notice also that closing the `gap' between \nctwo and \nlogspace in Theorem~\ref{thm:universality}\eqref{it:d-vass-unary-nctwo} would imply in particular solving the corresponding problem for UFA, which is an open question.

We observe that, as a corollary, we obtain procedures for testing the equivalence problem between an unambiguous VASS and a regular language. Indeed, the language of an unambiguous VASS $\A$ is equal to a regular language $L$ if, and only if, the VASS $\B$ resulting from the union of $\A$ and the DFA corresponding to the complement of $L$ is unambiguous and universal.

\paragraph*{Organization} 
We will prove Theorem~\ref{thm:universality} in Section~\ref{sec:universality} and Theorem~\ref{thm:testing-unambiguity} in Section~\ref{sec:unambiguity}. Each of these sections is divided into an ``upper bounds'' and ``lower bounds'' subsections. 
For reference, the upper and lower bounds of item \eqref{it:vass-expspace} of Theorem~\ref{thm:universality} are shown in Propositions~\ref{prop:expspace} and \ref{prop:univ-expspace-hard} respectively; item \eqref{it:d-vass-unary-nctwo}  in Propositions~\ref{prop:nctwo} and \ref{prop:nl-hard-univ}; item \eqref{it:d-vass-binary-pspace}  in Propositions~\ref{prop:pspace} and \ref{prop:pspace-hard-univ}; and item \eqref{it:ocn-binary-conp} in Proposition~\ref{prop:coNP-hardness}. 
The upper and lower bounds of item \eqref{it:unambcheck:vass-expspace} of Theorem~\ref{thm:testing-unambiguity} are shown in Propositions~\ref{prop:unamb-check} and \ref{prop:unamb-check-VASS-expspace-hard} respectively; item \eqref{it:unambcheck:d-vass-logspace}  in Propositions~\ref{prop:unamb-check} and~\ref{prop:unamb-check-dVASS-nl-hard}; item \eqref{it:unambcheck:d-vass-bin-pspace}  in Propositions~\ref{prop:unamb-check} and \ref{prop:unambcheck:pspace-hard}; and item \eqref{it:unambcheck:ocn-bin-conp} in Proposition~\ref{prop:conp-hard-unamb-checking}.
\section{Testing for Universality}
\label{sec:universality}
In this section we will prove Theorem~\ref{thm:universality}. Most of the section will be dedicated to proving the \expspace upper bound of item \eqref{it:vass-expspace}.


\subsection{Upper bounds}

\begin{proposition}[Theorem~\ref{thm:universality}\eqref{it:vass-expspace} upper bound]\label{prop:expspace}
The universality problem for unambiguous VASSes is in \expspace.
\end{proposition}

The proof strategy is as follows. First, we define an abstraction of a configuration, called an $N$-profile, for $N \in \N$, which is the result of replacing every number bigger than or equal to $N$ with $N$ in a configuration. The intuition is that any number bigger or equal $N$ is so big that we can disregard its exact value.
We next show that in certain circumstances, for any unambiguous $d$-VASS $V$ with $n$ states two configurations having equal $f(|V|,d,n)$-profile have also the same language, where $f$ is some fixed doubly-exponential function. This fact allows us to construct an unambiguous finite automaton $\A$ of doubly-exponential size, whose every state corresponds to one $f(|V|, d, n)$-profile, and such that $\A$ is universal if, and only if, $V$ is universal. As universality of UFAs is in \nctwo and therefore in \polylogspace, this gives us an \expspace algorithm for checking universality.

For any number $N \in \N$, the \defstyle{$N$-profile} of a configuration $(q,\bar v) \in Q \times \N^d$ is the pair $(q, \min(\bar v, N \cdot \bar 1))$.
Let $B_{M,d,n} = M \cdot A_{M,2d,2n^2}$, and let $C_{M,d,n} = M \cdot (B_{M,d,n} + 1)^d$.

We start with a useful lemma which bounds the length of runs witnessing ambiguity.
\begin{lemma}\label{lem:short-ambiguity-witness}
Let $V$ be a $d$-VASS with norm $M$ and $n$ states.
If $V$ is ambiguous then there exist two different runs accepting the same word of length at most $A_{M, 2d,2n^2}$ each.
\end{lemma}
\begin{proof}
Consider the following $2d$-VASS $V'$, which accepts exactly these words, which have at least two different accepting runs from the initial configuration of $V$.
The VASS $V'$ guesses two different runs of $V$ and simulates them, it is quite similar to a synchronized product of $V$ with itself.
In its $2d$ counters $V'$ keeps counter valuations of two configurations of $V$ of the simulated runs.
State of $V'$ is a pair of states of $V$ together with one bit of information indicating whether the two simulated runs have already differed or they are the same
till that moment. VASS $V'$ accepts if states of both simulated runs are accepting and the bit indicates that they have differed (even if now they are in the same state).
It is easy to see that $V'$ indeed accepts words, which have two different accepting runs in $V$.
Therefore if $V$ is ambiguous then $\lang(V')$ is nonempty.
Notice that the norm of $V'$ is bounded by $M$, as the norm of $V$ is. Therefore by Proposition~\ref{prop:rackoff} if $\lang(V')$ is nonempty then there is an accepting run of $V'$
of length at most $A_{M, 2d, 2n^2}$. Notice that the existence of such a run implies the existence of two different runs of $V$ over the same word, which additionally
also have length bounded by $A_{M, 2d, 2n^2}$. This finishes the proof.
\end{proof}

We state two basic properties of VASS which will be useful throughout.

\begin{claim}\label{claim:profile-language}
For any two configurations $c$ and $c'$ of a VASS $V$ with equal $(|V|\cdot N)$-profile, if $\rho$ is an accepting run from $c$ of length at most $N$ then $\rho$ is also accepting from $c'$.
\end{claim}

\begin{claim}[language monotonicity]\label{claim:monotonicity}
If $q(\bar u)$ and $q(\bar v)$ are two configurations of a VASS $V$ with $\bar u \preceq \bar v$ then $\lang(V, q(\bar u)) \subseteq \lang(V, q(\bar v))$.
\end{claim}

The following is the key lemma which will enable the improved complexity for the universality problem.
\begin{lemma}\label{lem:profile-reachable-ext}
Let $V$ be a universal, unambiguous $d$-VASS with $n$ states.
Then, any two configurations with equal $B_{|V|,d,n}$-profile reachable from the initial configuration have the same set of accepting runs (in particular, they have the same language).
\end{lemma}
\begin{proof}
By means of contradiction, let $c_1, c_2 \in Q \times \N^d$ be two configurations reachable from the initial configuration $c_\init$ with the same $B_{|V|,d,n}$-profile, but different sets of accepting runs.
Let $\rho$ be an accepting run from $c_1$ but not from $c_2$, reading the word $w$.

Let $c_\init \trans{u} c_2$. The word $uw$ is accepted by $V$ since it is universal, so there must be a configuration $c'_2$ such that
$c_\init \trans{u} c'_2$ and $w \in \lang(V, c'_2)$. Therefore $w$ is accepted both from configuration $c_1$ with the run $\rho$ and from configuration $c'_2$ with some accepting run $\hat\rho$. There are two cases to consider: either (i) $c_2 \neq c'_2$, or (ii) $c_2 = c'_2$ and $\hat\rho \neq \rho$.

For (i), let us first consider an (ambiguous) VASS $\tilde V$, being the result of adding $\eps$-labelled self-loops with effect $\bar 0$ in every state to $V$. 
Clearly, for every configuration $c$ we have $\lang(V, c) = \lang(\tilde V, c)$.
Let us confider a $2d$-VASS $V'$, which is a synchronized product of $\tilde V$ with itself: transitions, initial and accepting states are defined in a natural way.
Product is synchronized, so for any $a \in \Sigma_\eps$ there is an $a$-labelled transition in the product $V'$ iff there exist $a$-labelled transitions in the two components, both identical with $\tilde V$.
For two configurations $c = q(\bar u)$ and $c' = q'(\bar u')$ of $V$ we denote by $\lang(V', c, c')$ the language $\lang(V', (q,q')(\bar u, \bar u'))$.
Notice that, by construction, $\lang(V', c, c')$ is the intersection of $\lang(V, c)$ and $\lang(V, c')$. Therefore the word $w$ belongs to $\lang(V', c_1, c'_2)$.
By Proposition~\ref{prop:rackoff} there exists an accepting run $\rho'$ of $V'$ of length at most $A_{|V|,2d,n^2}$ reading a word $w'$ from $\lang(V', c_1, c'_2) = \lang(V, c_1) \cap \lang(V, c'_2)$. Consider the projection $\rho_1$ of $\rho'$ onto the first copy of $\tilde V$. We know thus that $\rho_1$ is accepting from $c_1$. Further, the  absolute value of the effect of $\rho_1$ on every coordinate is at most $|V| \cdot A_{|V|,2d,n^2} \leq |V| \cdot A_{|V|,2d,2n^2} = B_{|V|,d,n}$.
Recall that $c_1$ and $c_2$ have the same $B_{|V|,d,n}$-profile, so by~\cref{claim:profile-language} if $\rho_1$ is accepting from $c_1$ then it is also accepting from $c_2$.
Therefore $w' \in \lang(V, c_2)$ and $w' \in \lang(V, c'_2)$, which means that there are two distinct accepting runs over $u w'$ in $V$, contradicting the fact that it is unambiguous.

For (ii), we have that there are two distinct accepting runs for $w$ from $\max(c_1,c_2)$, namely $\rho$ and $\hat\rho$. Then, by Lemma~\ref{lem:short-ambiguity-witness}, there exist two different runs $\rho_1$ and $\rho_2$ from $\max(c_1,c_2)$ of length at most $A_{|V|,2d,2n^2}$ accepting the same word $w'$. Since $c_1$ and $c_2$ have the same $B_{|V|,d,n}$-profile, where $B_{|V|,d,n} = |V| \cdot A_{|V|,2d,2n^2}$, by Claim~\ref{claim:profile-language} both $\rho_1$ and $\rho_2$ are accepting from configuration $c_2$, and thus there are two distinct accepting runs over $uw'$ in $V$, contradicting the fact that it is unambiguous.
\end{proof}


\begin{corollary}\label{cor:profile-reachable}
If a universal, unambiguous $d$-VASS $V$ with $n$ states contains an accepting run with two configurations $c_1$ and $c_2$ such that $c_1$ occurs before $c_2$, then 
\begin{enumerate}[(i)]
	\item \label{cor:it:profiles-increment} if $c_1$ and $c_2$ have equal $B_{|V|,d,n}$-profile, then $c_1 \preceq c_2$;
	\item \label{cor:it:no-big-decrements} for every $i \in [1,d]$, $c_1[i] - c_2[i] < C_{|V|,d,n}$.
\end{enumerate}
\end{corollary}
\begin{proof}
\eqref{cor:it:profiles-increment} By means of contradiction, let $c_1$ and $c_2$ be configurations with the same profile such that $c_1 \not\preceq c_2$, meaning that $c_1[i] > c_2[i]$ for some $i$.
Let $\rho_1\rho_2\rho_3$ be an accepting run of $V$, such that $\rho_1$ reaches the configuration $c_1$ from the initial configuration, and $\rho_2$ reaches the configuration $c_2$ from configuration $c_1$. Since the effect of $\rho_2$ decrements component $i$, it is easy to see that there is some $k \in \N$ such that $(\rho_2)^k\rho_3$ is an accepting run from $c_1$ but not from $c_2$, contradicting Lemma~\ref{lem:profile-reachable-ext} above.

\eqref{cor:it:no-big-decrements} Suppose there is a decrement of at least $C_{|V|,d,n}$ at some coordinate $i$. Since $C_{|V|,d,n} = |V| \cdot (B_{|V|,d,n} + 1)^d$ is at least the number of
$B_{|V|,d,n}$-profiles times the biggest effect of a transition, this means that at least $k=B_{|V|,d,n}$ distinct configurations $c'_1, \dotsc, c'_k$  occur in the run between $c_1$ and $c_2$ such that $c_1[i] > c'_1[i]>c'_2[i] > \dotsb > c'_k[i]$. Hence, among $c_1, c'_1, \ldots, c'_k$ there must be two equal $B_{|V|,d,n}$-profile configurations, contradicting the item \eqref{cor:it:profiles-increment} above.
\end{proof}

This last statement can be informally understood as follows: if $V$ is universal, then it is still universal if configurations are abstracted by their $C_{|V|,d,n}$-profiles. We now formalize what this means.
\newcommand{\toProf}[1]{\lfloor #1 \rfloor}%
Let us fix an unambiguous VASS $V$, and let us henceforth write $\omega$ as short for $C_{|V|,d,n}$. For any configuration $c$ let $\toProf{c}$ denote its $\omega$-profile, that is, $\toProf{q(\bar u)} = q(\min(\bar u, \omega \cdot \bar 1))$.
Let $V = (\Sigma, d ,Q_V, q_V, \delta_V, F_V)$ be an unambiguous VASS.
We construct a finite automaton $\A_V = (\Sigma, Q_\A, q_\A, \delta_A, F_\A)$ in the following way:
\begin{itemize}
  \item the set of states $Q_\A$ is the set of pairs $Q_V \times [0, \omega]^d$;
  \item the initial state $q_\A$ is $q_V(\bar 0)$;
  \item the set of final states $F_\A$ consists of all the pairs having the first coordinate in $F_V$, namely $F_\A = F_V \times [0, \omega]^d$;
  \item $\delta_\A$ is the set of all transitions $p(\bar u) \trans{a} q(\toProf{\bar u+\bar v})$ such that $(p, a, \bar v, q) \in \delta_V$ and $\bar u + \bar v \in \N^d$.
\end{itemize}

We now show that $\A_V$ is unambiguous, and that it is universal iff $V$ is universal.
\begin{lemma}\label{lem:V-simulates-AV}
	For every run $p_1(\bar u_1) \trans{a_1} p_2(\bar u_2) \trans{a_2} \dotsb p_{n}(\bar u_{n}) \trans{a_{n}} p_{n+1}(\bar u_{n+1})$  of $\A_V$ there is a run $(p_1,a_1,\bar v_1,p_2) \dotsb (p_{n},a_{n},\bar v_{n},p_{n+1})$ of $V$ such that $\bar v_1 + \dotsb + \bar v_i \geq \bar u_i$ for every $i \in [1,n]$.
\end{lemma}
\begin{proof}
This can by shown by induction on $n$. It suffices to replace every transition $p_i(\bar u_i) \trans{a_i} p_{i+1}(\bar u_{i+1})$ of $\A_V$ by a transition $(p_i, a_i, \bar v, p_{i+1}) \in \delta_V$ such that $\bar u_{i+1} = \toProf{\bar u_i+\bar v}$, which exists by construction.
\end{proof}
As a consequence of the previous lemma, if there are two distinct accepting runs for a word $w$ in $\A_V$, then there are also two distinct accepting runs over $w$ in $V$. In other words:
\begin{lemma}
If $V$ is unambiguous then $\A_V$ is unambiguous.
\end{lemma}
%

\begin{lemma}\label{lem:language-equivalence}
$V$ is universal if, and only if, $\A_V$ is universal.
\end{lemma}
\begin{proof}
Observe first that $\lang(\A_V) \subseteq \lang(V)$ by Lemma~\ref{lem:V-simulates-AV}.
Hence, if $\A_V$ is universal, so is $V$.
For the converse direction, suppose $V$ is universal, and let us show that $\A_V$ is universal as well. 
Let $\rho = (q_0,a_1,\bar v_1,q_1) \dotsb (q_{n-1},a_n,\bar v_n,q_n)$ be the accepting run of $w = a_1 \dotsb a_n$ in $V$. Let us consider the run $\rho' = (q_0(\bar x_0),a_1, q_1(\bar x_1)) \dotsb (q_{n-1}(\bar x_{n-1}),a_n, q_n(\bar x_n))$ of $\A_V$, where $\bar x_0 = \bar 0$ and for every $i>0$, $\bar x_{i} = \toProf{\bar x_{i-1} + \bar v_i}$.
We claim that $\rho'$ is an accepting run on $\A_V$. By means of contradiction, if $\rho'$ is not a run, there must be some $q_{i}(\bar x_{i}) \trans{a_{i+1}} q_{i+1}(\bar x_{i+1})$ which is not a transition of $\A_V$. This can only happen if some configuration on $\rho$ reaches some big counter value at a position $j$ which later decreases by at least $\omega$. More concretely, this means that there are, among the configurations reachable through $\rho$, two configurations $c,c'$ such that $c$ appears before $c'$ and for some $j \in [1,k]$ we have $c[j] - c'[j] > \omega$. But this would contradict Corollary~\ref{cor:profile-reachable}-\eqref{cor:it:no-big-decrements}. Hence, $\rho'$ is an accepting run and thus $\A_V$ is universal.
\end{proof}

Notice that the automaton $\A_V$ has a doubly-exponential number of states. As checking its universality is polynomial-time in its size~\cite{DBLP:conf/dcfs/Colcombet15},
which is doubly exponential, the problem is in \twoexptime.
In order to design an \expspace algorithm we need a bit more work.
The following lemma together with Lemma~\ref{lem:language-equivalence} finishes the proof of Proposition~\ref{prop:expspace}.

\begin{lemma}\label{lemma:universality-automaton}
Checking universality of $\A_V$ is in \expspace.
\end{lemma}
\begin{proof}
Notice first that the function $V \mapsto \A_V$ can be easily computed in \expspace.
Indeed, a state of $\A_V$ is described by a pair consisting of a state from $Q_V$ and a vector $\bar v \in [0, \omega]^d$, where $\omega = C_{|V|,d,|Q_V|} = |V| \cdot ( |V| \cdot (4|Q_V|^4 (|V|+1)^2)^{(8d)^{2d-1}}+1)^d$ is doubly exponential with respect to the description size of $V$, and therefore it can be kept in \expspace. It is then possible to iterate through
all the possible pairs in $(Q_V, [0, \omega]^d)$ in \expspace and for every state output the transitions outgoing from this state.

By~\cite{DBLP:journals/ipl/Tzeng96} checking universality of UFA without cycles containing only $\eps$-labelled transitions
($\eps$-cycles) is in \nctwo, namely in the class of languages
recognizable by uniform families of circuits of depth $\O(\log^2(n))$ and binary branching, where $n$ is the number of inputs.
A simple procedure which eliminates all the $\eps$-cycles (\ie, all the transitions involved in $\eps$-cycles) can be designed to be in $\nl$. Observe that eliminating $\eps$-cycles does not change the language of unambiguous automata, since no accepting run can contain a transition from an $\eps$-cycle (such a run extended by the $\eps$-cycle would be also accepting, which would violate the unambiguity assumption).
Since $\nl \subseteq$ \nctwo and \nctwo is closed under composition, we obtain that the universality problem for an arbitrary UFA (possibly with $\eps$-transitions) is in \nctwo as well. 
It is folklore that \nctwo is included in poly-logarithmic space (actually in the deterministic space $\log^2 n$).
Indeed, one can simply simulate a circuit of depth $D$ and binary branching in space $D$.

It is now enough to argue that the composition of \expspace and \polylogspace is included in \expspace.
This result is also folklore, we sketch here a proof. Any algorithm in the composition of \expspace and \polylogspace can be
seen as a \polylogspace algorithm inputting the output of an \expspace machine, potentially of a doubly exponential length.
This doubly exponential output cannot be kept by an \expspace algorithm, but one can simulate the composition by
a \polylogspace algorithm asking \expspace oracles for particular letters of its input.
Such an algorithm in turn can be simulated easily in \expspace. We keep three exponential size pieces of the information:
(i) the space of the oracle, (ii) the index of the doubly exponential input being currently transferred to the oracle,
and (iii) the space of the poly-logarithmic algorithm, which is poly-logarithmic with respect to the doubly exponential input, hence exponential.
Therefore indeed $\expspace \circ \polylogspace \subseteq \expspace$, which finishes the proof.
\end{proof}

Let us now analyze the situation for a fixed dimension $d \in \N$.
The number of states of $\A_V$ equals $|Q_V|$ times $|V| \cdot ( |V| \cdot (4|Q_V|^4 (|V|+1)^2)^{(8d)^{2d-1}}+1)^d$, which for a fixed $d$
is a polynomial depending on $|Q_V|$ and $|V|$. This immediately implies that for $V$ represented in unary the size of $\A_V$
is polynomial, while for $|V|$ represented in binary the size of $\A_V$ is exponential in the size of the input.
A proof almost identical to that of Lemma~\ref{lemma:universality-automaton}, where we substitute \expspace with \pspace, yields the following result.
\begin{proposition}[Theorem~\ref{thm:universality}\eqref{it:d-vass-binary-pspace} upper bound]\label{prop:pspace}
For every fixed $d \in \N$ the universality problem for binary represented, unambiguous $d$-VASS is in \pspace.
\end{proposition}

In a similar way we solve the case of unary represented $d$-VASSes.
In this case, we replace \expspace with the class of problems solvable in logarithmic space \logspace.
We also use the fact that \logspace composed with \nctwo is included in \nctwo, which is immediately implied by a trivial
closure of \nctwo by composition and inclusion \logspace $\subseteq$ \nctwo. Then we get the following.

\begin{proposition}[Theorem~\ref{thm:universality}\eqref{it:d-vass-unary-nctwo} upper bound]\label{prop:nctwo}
For every fixed $d \in \N$ the universality problem for unary represented, unambiguous $d$-VASS is in \nctwo.
\end{proposition}


\subsection{Lower bounds}

\begin{proposition}[Theorem~\ref{thm:universality}\eqref{it:vass-expspace} lower bound]\label{prop:univ-expspace-hard}
	The universality problem for unambiguous VASS is \expspace-hard, even on a one-letter alphabet.
\end{proposition}
\begin{proof}
We reduce from the problem of whether an unambiguous $\epsilon$-VASS has an empty language, which is \expspace-hard as observed in Lemma~\ref{lem:epsVASS-emptiness}.
Given an unambiguous $\epsilon$-VASS $\A=(\set a,d,Q,q_0,\delta,F)$, we build an unambiguous VASS $\B$ on a one-letter alphabet $\set a$ such that $\lang(\B) = a^*$ if $\lang(\A)=\set{\epsilon}$ and $\lang(\B) = \emptyset$ otherwise. $\B$ is the result of adding a new final state $q_f$ to $\A$, and transitions $(q,a,\bar 0, q_f)$ for every $q \in F \cup \set{q_f}$.
\end{proof}

\begin{corollary}
The co-finiteness problem for unambiguous VASS, that is, whether the complement of its language is finite, is \expspace-hard.
\end{corollary}

We leave open the question of whether the lower bound of Proposition~\ref{prop:univ-expspace-hard} still holds for unambiguous VASS without epsilon transitions.

The following proposition proves the lower bound of Theorem~\ref{thm:universality}\eqref{it:d-vass-binary-pspace}.
\begin{proposition}[Theorem~\ref{thm:universality}\eqref{it:d-vass-binary-pspace} lower bound]\label{prop:pspace-hard-univ}
	The universality problem for unambiguous 2-VASS is \pspace-hard.
\end{proposition}
\begin{proof}
We reduce from the bounded one-counter automata reachability problem, which is known to be \pspace-hard \cite[Corollary~10]{FearnleyJ15}. This problem can be stated as follows: given  a 1-VASS $\A = (\Sigma,1,Q_\A,q,\delta_\A,F)$, a number $N \in \N$ encoded in binary, and a configuration $p(m)$, is there a run $(r_1,\alpha_1,u_1, r'_1) \dotsb (r_n,\alpha_n,u_n, r'_n)$ from $q(0)$ to $p(m)$ such that $\sum_{i \leq j} u_i \leq N$ for every $1 \leq j \leq n$?
The alphabet is not important for this problem, we can consider that every transition reads the letter $a$.

Let $\A$, $N$, $p(m)$ be the input of the aforementioned problem. We now construct, in polynomial time, an unambiguous 2-VASS $\B = (\Sigma,2,Q,q_0,\delta,F)$, such that it is universal if, and only if, the answer to the input is negative --- the statement then follows by closure under complement of \pspace.
Concretely, the language of $\B$ is essentially the set of all sequences of transitions in $(\delta_\A)^*$ which \emph{do not} contain a run from $q(0)$ to $p(m)$ as a prefix.
Intuitively, the construction of $\B$ from $\A$ can be divided into two steps. First we change the $N$-bounded 1-VASS into a 2-VASS by simulating configuration $q(i)$ by $q(i, N-i)$.
However, this 2-VASS might be far from being universal. Therefore, we add to it a lot of transitions such that it is almost universal: the only way for a word not to be accepted is to reach a configuration corresponding to $p(m)$.

The construction of $\B$ is as follows.
The alphabet $\Sigma$ is defined as $\delta_{\A} \dcup \set\star$; the state set $Q$ is defined as $Q_\A \dcup \set{\bot,q_f,q_0}$; and the set of final states is $F = Q \setminus \set\bot$, where $\bot$ is a \emph{sink state}.
$\B$ will always keep the invariant that the sum of its two components is equal to $N$ on all configurations with state in $Q_\A$ reachable from the initial configuration $q_0(0,0)$. Further, the transition graph is as in $\delta_\A$ but labels are used to enforce unambiguity. This is done by initializing the vector in $(0,N)$ as the first thing the automaton does (by adding a new initial state $q_0$
and transition $(q_0,t,(0,N),q)$ from it to the initial state of $\A$), and additionally translating every transition $t=(r,a,h,r') \in \delta_\A$ into $(r,t,(h,-h),r')$.
Now we need to assure that the only way to be not accepted is to reach configuration $p(m, N-m)$.
For that purpose we add a special transition reading $\star$ with effect $(-m,m-N)$ and going from $p$ to the sink state $\bot$.
All the other sequences of transitions need to be made accepting.
For that we add an extra accepting state $q_f$ and a lot of transitions leading to it.
Concretely, $\B$ has these transitions:
\begin{enumerate}[(i)]
  \item the initial transition $(q_0,t,(0,N),q)$ for every $t \in \Sigma$;
  \item \label{it:simulatingtransition} a `simulating' transition $(r,t,(h,-h),r')$ for every $t=(r,a,h,r') \in \delta_\A$;
  \item a transition from $p$ to $\bot$ reading $\star$ with effect $(-m, m - N)$;
  \item a transition from every $r \in Q_\A \setminus \set p$ to $q_f$ reading $\star$ with effect $(0,0)$;
  \item two transitions from $p$ to $q_f$ reading $\star$, one with effect $(-m-1,0)$ and one with effect $(0,-(N-m)-1)$;
  \item a transition from every $r \in Q_\A$ to $q_f$ reading $t \in \Sigma \setminus \set \star$ with effect $(0,0)$ if the $t$-labelled transition is not outgoing from $r$;
  \item \label{it:specialtransition} a transition from every $r \in Q_\A$ to $q_f$ reading $t \in \Sigma \setminus \set \star$ with effect $(0,-(N+\ell)-1)$ if the $t$-labelled transition is outgoing from $r$ and has effect $\ell < 0$;
  \item $\bar 0$-effect self-loops on $q_f$, with all possible letters of $\Sigma$.
\end{enumerate}
Figure \ref{fig:pspace-hard} contains a depiction of the construction.
\begin{figure}
	\centering
		\includegraphics[width=\textwidth]{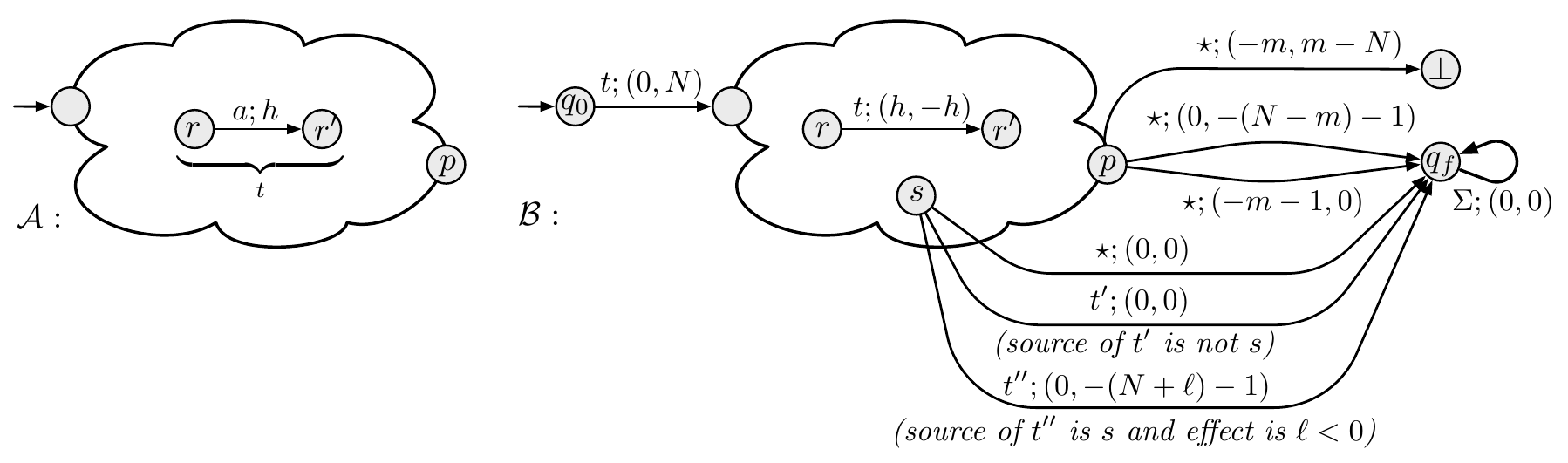}
	\caption{Definition of $\B$. An arrow labelled ``$\alpha;\bar x$'' denotes a transition reading $\alpha$ with effect $\bar x$.}
	\label{fig:pspace-hard}
\end{figure}
We now show the correctness of the reduction.
Observe first that, by construction, all configurations $c$ of $\B$ reachable from $q_0(0,0)$ are $N$-bounded. Further, if the state of $c$ is from $Q_\A$, then the sum of its components is equal to $N$.
	
		\smallskip
	
		We show that $\B$ is unambiguous. What is more, we will show that for every configuration $r(\bar u_0)$ reachable from $q_0(0,0)$ and for every letter $a \in \Sigma$ there is at most one outgoing transition from $r$ reading $a$ that can be applied to $r(u_0,u'_0)$. By means of contradiction, suppose that there are two distinct transitions $(r,a,(u_1,u'_1),r_1),(r,a,(u_2,u'_2),r_2) \in \delta$ such that $(u_0,u'_0)+(u_1,u'_1) \in \N^2$ and $(u_0,u'_0) + (u_2,u'_2) \in \N^2$. By construction, the only possibility is that one transition is a simulating transition as defined in \eqref{it:simulatingtransition}, and the other transition is as defined in \eqref{it:specialtransition}. In particular, $a$ must be a transition from $\delta_\A$, $r,r_1$ are states from $Q_\A$, and $r_2 = q_f$. By the above observation, $u_0+u'_0 = N$, and by construction (item \eqref{it:specialtransition}), $u_1<0$ and $u'_2 = -(N + u_1) -1$.
Since $u'_0 + u'_2 \geq 0$ by the hypothesis $(u_0,u'_0) + (u_2,u'_2) \in \N^2$, we can replace $u'_2$ with the equality $u'_2 = -(N + u_1) -1$ just observed, and we obtain $u'_0 -(N + u_1) -1 \geq 0$. Since we also know that $u_0+u'_0=N$ by the observation above, we can further replace $u'_0$ with $N-u_0$ in $u'_0 -(N + u_1) -1 \geq 0$, and we obtain $u_0 + u_1 <0$. Note that this contradicts the hypothesis $(u_0,u'_0)+(u_1,u'_1) \in \N^2$. The contradiction comes from assuming that both transitions were possible to trigger.
		\smallskip
	
		We finally show that $\B$ is universal if, and only if, there is no $N$-bounded run from $q(0)$ to $p(m)$ in $\A$. 
	Observe first that for every word $w \in \Sigma^*$ there is exactly one run of $\B$ reading $w$.
		If there is an $N$-bounded run $\rho$ from $q(0)$ to $p(m)$ in $\A$, it follows that the run of $\B$ reading $\rho \star$ ends in the sink state $\bot$, and thus $\rho \star \not \in \lang(\B)$, witnessing the fact that $\B$ is not universal.
		If, on the other hand, there is a run of $\B$ ending in state $\bot$, it must be reading a word of the form $\rho \star$ where $\rho$ is an $N$-bounded run from $q(0)$ to $p(m)$ in $\A$. Since $\bot$ is the sole state which is not accepting and since, as observed before, for all words there is a run, it follows that if $\B$ is universal, then there is no $N$-bounded run in from $q(0)$ to $p(m)$ in $\A$.
\end{proof}

Finally, we show \conp-hardness for universality of one counter nets.
\begin{proposition}[Theorem~\ref{thm:universality}\eqref{it:ocn-binary-conp}]\label{prop:coNP-hardness}
The universality problem for unambiguous 1-VASS is \conp-hard.
\end{proposition}
\begin{proof}
We equivalently will show that non-universality problem for unambiguous 1-VASSes is \np-hard.
The reduction is from the \partition problem. In the \partition problem we are given a finite set of natural numbers $S = \{n_1, \ldots, n_k\} \subseteq_\fin \N$
and we are supposed to answer whether the set of indices $[1,k]$ can be partitioned into two subsets $I_1, I_2 \subseteq [1,k]$ such that
$\sum_{i \in I_1} n_i = \sum_{i \in I_2} n_i$. Such a partition is called a \emph{perfect partition}.
All the numbers are binary represented. The \partition problem is known to be \np-hard~\cite{DBLP:conf/coco/Karp72}.

For an instance of a \partition problem $S \subseteq_\fin \N$ we build an unambiguous 1-VASS $V_S$ such that perfect partition for $S$ exists
if and only if the 1-VASS is not universal. Let $\sum_{i \in [1,k]} n_i = N$, note that the perfect partition exists iff there is a set of indices $I$ such that $\sum_{i \in I} n_i = N / 2$.
Every word of length $k$ encodes a natural number $S_w$ in the following way: for $w \in \{0, 1\}^k$ we define $S_w = \sum_{i \mid w[i]=1} n_i$.
We will design $V_S$ in such a way that $\lang(V_S) \subseteq \{0, 1\}^*$ will always contain all the words of length different than $k$. 
Among words of length $k$ language $\lang(V_S)$ will contain exactly these for which $S_w \neq N / 2$.
Then indeed $\lang(V_S)$ would be not universal iff set $S$ has a perfect partition.

The 1-VASS $V_S$ is defined in Figure~\ref{fig:conp-vass}. 
\begin{figure}
	\centering
		\includegraphics[width=.95\textwidth]{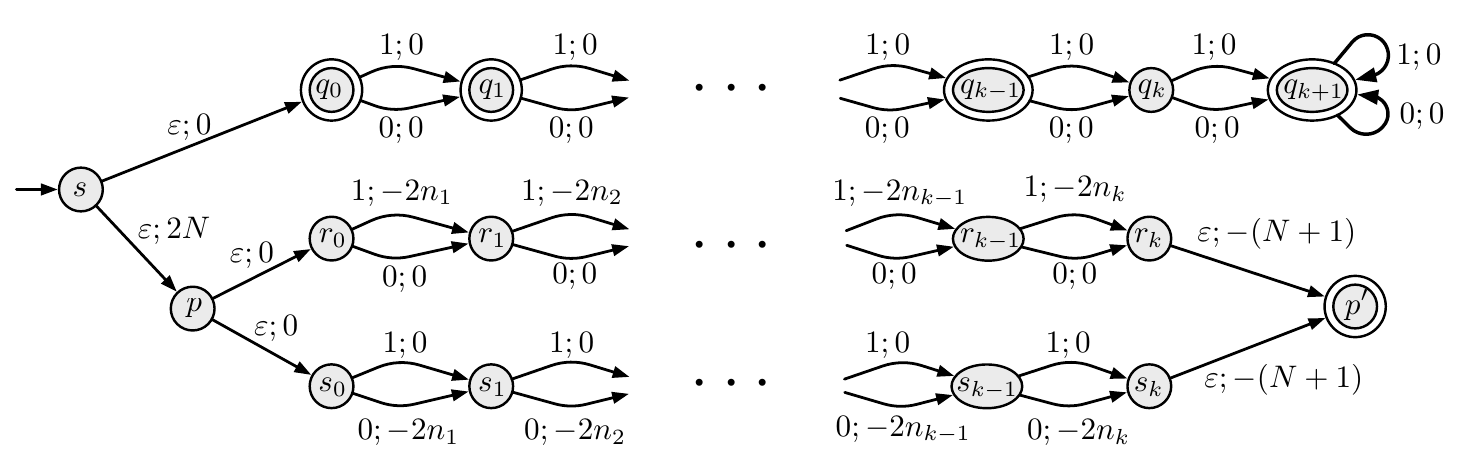}
	\caption{Definition of $V_S$. An arrow labelled ``$i;\ell$'' denotes a transition reading $i$ with effect $\ell$. Double circled states are final.}
	\label{fig:conp-vass}
\end{figure}
It consists of two parts: the top part, with states $q_0$ to $q_{k+1}$ accepts all the words of length different than $|S| = k$, while the bottom part accepts some words of length $k$.

It is immediate to see that the top part accepts all words of length different to $k$, and all of them by exactly one run.

The bottom part consists of states: $p$, $p'$, $r_0, r_1, \ldots, r_k$ and $s_0, s_1, \ldots, s_k$, where  only the state $p'$ is accepting.
Notice that transitions in states $r_i$ are mirrored with respect to transitions in states $s_i$, namely effect of a transition over some letter from $r_i$
equals the effect of the transition over the other letter in $s_i$.
Let us inspect now how an accepting run over $w \in \{0, 1\}^k$ can look like.
Every such run starts from $q_0(0)$ and then goes to $p(2N)$.
Then it splits into two runs, to $r_0(2N)$ and $s_0(2N)$ and from this moment on there are two runs: one in some state $r_i$ and the other in the corresponding state $s_i$.
Then after reading the whole $w$ the two runs are in configurations $r_k(2N - 2S_w)$ and $s_k(2N - 2(N-S_w)) = s_k(2S_w)$.
Notice that $0 \leq S_w \leq N$, so both configurations are indeed always reachable. 
Now comes the last transition from either $r_k$ or $s_k$ to $p'$. Observe that if $S_w \neq N / 2$ then exactly one of them
can be fired. Indeed if $S_w \neq N / 2$ so $2S_w \neq N$ then exactly one of the numbers $2S_w$ and $2N - 2S_w$ equals at least $N+1$.
Then from exactly one of the configurations $r_k(2N - 2S_w)$ and $s_k(2S_w)$ counter value $N+1$ can be subtracted and the run over the word
$w$ will reach an accepting configuration $p'(c)$ for some $c \geq 0$. Then we have $w \in \lang(V_S)$ and exactly one accepting run over $w$.
On the other hand assume now that $S_w = N / 2$. Then the two reached configurations
are $r_k(N)$ and $s_k(N)$. In none of them counter value $N+1$ can be subtracted, which means that in that case no accepting run
over $w$ exists and $w \not\in \lang(V_S)$. Therefore indeed $V_S$ is unambiguous
and importantly $\lang(V_S)$ is not universal iff there exists a perfect partition for $S$.
This finishes the proof.
\end{proof}

\begin{proposition}[Theorem~\ref{thm:universality}\eqref{it:d-vass-unary-nctwo} lower bound]\label{prop:nl-hard-univ}
The universality problem for $d$-VASS with unary encoding is \nlogspace-hard, for every $d \geq 1$.
\end{proposition}
\begin{proof}
	This already holds for UFA.
\end{proof}

\section{Testing for Unambiguity}
\label{sec:unambiguity}
Here we will prove Theorem~\ref{thm:testing-unambiguity}. As we will see, upper bounds follow from the emptiness problem and lower bounds from adaptations of the reductions from the previous section.


\subsection{Upper bounds}
We will next prove the upper bound of Theorem~\ref{thm:testing-unambiguity}\eqref{it:unambcheck:vass-expspace},
\eqref{it:unambcheck:d-vass-logspace} and \eqref{it:unambcheck:d-vass-bin-pspace} namely:

\begin{proposition}[Theorem~\ref{thm:testing-unambiguity}\eqref{it:unambcheck:vass-expspace},
\eqref{it:unambcheck:d-vass-logspace} and \eqref{it:unambcheck:d-vass-bin-pspace} upper bound]\label{prop:unamb-check}
The unambiguity checking problem is:
\begin{enumerate}[(i)]
  \item in \expspace for VASSes with binary encoding;
  \item in \nl for $d$-VASSes with unary encoding for any fixed $d \in \N$;
  \item in \pspace  for $d$-VASSes with binary encoding for any fixed $d \in \N$.
\end{enumerate}
\end{proposition}

\begin{proof}
By Lemma~\ref{lem:short-ambiguity-witness} if a $d$-VASS with norm $M$ and $n$ states is ambiguous then there exists two different runs of length at most $A_{M, 2d, 2n^2}$ accepting the same word.
Number $A_{M, 2d, 2n^2} = (4n^4 (M+1)^2)^{(8d)^{2d-1}}$ is doubly exponential wrt. the size of the VASS representation when $M$ is given in binary and $d$ is not fixed.
For fixed $d$ an $M$ given in binary $A_{M, 2d, 2n^2}$ is exponential wrt. the input and for fixed $d$ and $M$ given in unary it is polynomial wrt. the input.
Therefore the algorithm, which enumerates all the pairs of different runs of length up to $A_{M, 2d, 2n^2}$ and checks whether some pair accepts the same word
works in \expspace, \pspace and \nl, respectively, which finishes the proof.
\end{proof}

\subsection{Lower bounds}

\begin{proposition}[Theorem~\ref{thm:testing-unambiguity}\eqref{it:unambcheck:vass-expspace} lower bound]\label{prop:unamb-check-VASS-expspace-hard}
The unambiguity checking problem for VASS with unary encoding is \expspace-hard.
\end{proposition}
\begin{proof}
We reduce from the problem of whether an unambiguous $\eps$-VASS has an empty language,
which is \expspace-complete as mentioned in Lemma~\ref{lem:epsVASS-emptiness} (it is a consequence of Lipton's construction~\cite{lipton1976}).
To an unambiguous $\eps$-VASS we add one state accepting the empty word $\eps$.
Then the constructed VASS is unambiguous iff the original one has empty language, which finishes the \expspace-hardness proof.
\end{proof}

\begin{proposition}[Theorem~\ref{thm:testing-unambiguity}\eqref{it:unambcheck:d-vass-logspace} lower bound]\label{prop:unamb-check-dVASS-nl-hard}
The unambiguity checking problem for $d$-VASS with unary encoding is \nlogspace-hard.
\end{proposition}
\begin{proof}
This is already true for finite automata.
\end{proof}

\begin{proposition}[Theorem~\ref{thm:testing-unambiguity}\eqref{it:unambcheck:d-vass-bin-pspace} lower bound]\label{prop:unambcheck:pspace-hard}
The unambiguity checking problem for 2-VASS is \pspace-hard.
\end{proposition}
\begin{proof}
	This is a corollary of the construction in the proof of Proposition~\ref{prop:pspace-hard-univ}. One can adapt the automaton by now having $\bot$ as a sole accepting state, and all other states as non-accepting, and adding a transition $(\bot, \epsilon, (0,0), \bot)$, in such a way that $\B$ is unambiguous if, and only if, there is no run that reaches $\bot$.
\end{proof}

\begin{proposition}[Theorem~\ref{thm:testing-unambiguity}\eqref{it:unambcheck:ocn-bin-conp}]\label{prop:conp-hard-unamb-checking}
	The unambiguity checking problem for $1$-VASS is \conp-hard.
\end{proposition}
\begin{proof}
A construction very similar to the one used to show \conp-hardness of universality  (Proposition~\ref{prop:coNP-hardness}) can be used to show that unambiguity checking for 1-VASS is \conp-hard.
If instead of transitions $r_k \trans{\eps; -(N+1)} p'$ and $s_k \trans{\eps; -(N+1)} p'$ we 
have transitions $r_k \trans{\eps; -N} p'$ and $s_k \trans{\eps; -N} p'$, then $V_S$ is ambiguous if and only if
there is a perfect partition for $S$. This shows that ambiguity checking is \np-hard and unambiguity checking
is \conp-hard.
\end{proof}

\section{Discussion}

We leave open the question about the exact complexity of universality problem for unambiguous 1-VASS with transitions represented in binary, which we showed to be \pspace-easy and \conp-hard. We conjecture that it is \conp-complete.
Another question that we leave open is the complexity of the universality problem for VASS without $\eps$-transitions; our \mbox{\expspace}-hardness of Proposition~\ref{prop:univ-expspace-hard} crucially uses $\eps$-transitions, and it is not clear whether it can be adapted to avoid them.
We conjecture that the universality problem for unambiguous VASS without $\eps$-transitions is still \expspace-hard.
An open question related to the gap of Theorem~\ref{thm:universality}\eqref{it:d-vass-unary-nctwo} is the one about the precise complexity of the universality problem for unambiguous finite automata, which is \nl-hard and only known to be in \nctwo~\cite{DBLP:journals/ipl/Tzeng96}.


While we have focused our study on the universality and unambiguity checking problems for unambiguous VASS, we point out that there are many intriguing unanswered problems on unambiguous systems.
%
%
In particular, closely related to the universality problem are: co-finiteness, equivalence and inclusion problems.
The universality problem is often strongly connected with the equivalence and inclusion problems. As observed in Section~\ref{sec:results}, the techniques allow for answering the equivalence problem with a regular language. However, equivalence between two unambiguous VASS seems a more difficult question. In particular, observe that trying to reduce $L(\A) \subseteq L(\B)$ to $L((\A) \cap L(\B)) \cup \overline{L(\B)} = \Sigma^*$ would fail in this case, since VASS and unambiguous VASS are not closed under complement ---in fact, the only VASSes whose complement is a VASS are those denoting regular languages \cite{CzerwinskiLMMKS18}.

It is natural to ask about the decidability and complexity of these problems for most fundamental models of computation:
finite automata, one counter nets, VASS or even pushdown automata (PDA) under the assumption of unambiguity. We give some examples. While equivalence of VASS languages is undecidable, is it decidable for unambiguous VASS? Language equivalence is undecidable for PDA and decidable for deterministic PDA (by the celebrated result of Sénizergues~\cite{DBLP:journals/tcs/Senizergues01}), but might it still be decidable for unambiguous PDA? And what about universality?

\paragraph{Acknowledgements} We thank Lorenzo Clemente for leading us to the \nctwo membership for UFA universality problem.



\bibliography{citat}

\end{document}